\documentclass[smallextended]{svjour3}       % onecolumn (second format)
\smartqed  % flush right qed marks, e.g. at end of proof
\usepackage{graphicx}

\begin{document}

\title{A note on triangle-free graphs}


\author{V.V. Mkrtchyan \and
        P.A. Petrosyan
}


\institute{V.V. Mkrtchyan \at
                Department of Informatics and Applied Mathematics,\\
Yerevan State University, 0025, Armenia\\
              \email{vahanmkrtchyan2002@\{ysu.am, ipia.sci.am,
               yahoo.com\}}\\
             \and
              P.A. Petrosyan \at
              Department of Informatics and Applied Mathematics,\\
Yerevan State University, 0025, Armenia\\
              \email{pet\_petros@ipia.sci.am}\\
}

\date{Received: date / Accepted: date}

\maketitle

\begin{abstract}
We show that if $G$ is a simple triangle-free graph with $n\geq 3$
vertices, without a perfect matching, and having a minimum degree at
least $\frac{n-1}{2}$, then $G$ is isomorphic either to $C_5$ or to $K_{\frac{n-1}{2},\frac{n+1}{2}}$.\\


\end{abstract}

\section{Introduction}
\label{intro}

We consider finite undirected graphs that do not contain loops or
multiple edges. Let $V(G)$ and $E(G)$ denote the sets of vertices
and edges of $G$, respectively. The degree of a vertex $v\in V(G)$
is denoted by $d_{G}(v)$ and the diameter of $G$ by $diam(G)$. For a
graph $G$, let $\delta(G)$ and $\Delta(G)$ denote the minimum and
maximum degree of $G$, respectively. For $n\geq 3$, let $C_{n}$
denote the cycle of length $n$. The cycle $C_3$ is called a
triangle. For $m,n\geq 1$, let $K_{m,n}$ denote the complete
bipartite graph one part of which has $m$ vertices and the other
part $n$ vertices. Terms and concepts that we do not define can be
found in \cite{b6}.

It is well-known that triangle-free graphs play an important role in
graph theory. One of the first results concerns triangle-free graphs
is the following theorem of Mantel \cite{b4}.

\begin{theorem}
\label{mytheorem1} If $G$ is a simple triangle-free graph with $n$
vertices and $m$ edges, then
\begin{center}
$m\leq \left\lfloor \frac{n^{2}}{4}\right\rfloor$.
\end{center}
\end{theorem}

Note that the upper bound in Theorem \ref{mytheorem1} is sharp for
the complete bipartite graph $K_{\lfloor
\frac{n}{2}\rfloor,\lceil\frac{n}{2}\rceil}$. Moreover, $K_{\lfloor
\frac{n}{2}\rfloor,\lceil\frac{n}{2}\rceil}$ is a unique graph with
$\lfloor \frac{n^{2}}{4}\rfloor$ edges. Clearly, every bipartite
graph is a triangle-free graph. On the other hand, in 1974,
Andr\'asfai, Erd\H{o}s and S\'os \cite{b1} found the minimum degree
condition which forces a triangle-free graph to be bipartite.

\begin{theorem}
\label{mytheorem2} If $G$ is a simple triangle-free graph with $n$
vertices and $\delta(G)> \frac{2}{5}n$, then $G$ is bipartite.
\end{theorem}

Also, a similar result for triangle-free graphs was obtained by
Erd\H{o}s, Fajtlowits and Staton \cite{b2} in 1991.

\begin{theorem}
\label{mytheorem3} If $G$ is a simple triangle-free graph with no
three vertices having equal degree, then $G$ is bipartite.
\end{theorem}

In 1989, Erd\H{o}s, Pach, Pollack and Tuza \cite{b3} investigated
the connection between the diameter and minimum degree of connected
triangle-free graphs. In particular, they proved the following

\begin{theorem}
\label{mytheorem4} If $G$ is a connected triangle-free simple graph
with $n\geq 3$ vertices, and $\delta(G)\geq 2$, then
\begin{center}
$diam(G)\leq \left\lceil
\frac{n-\delta(G)-1}{2\delta(G)}\right\rceil$.
\end{center}
\end{theorem}

In this short note we prove that if $G$ is a simple triangle-free
graph with $n\geq 3$ vertices, without a perfect matching, and
having a minimum degree at least $\frac{n-1}{2}$, then $G$ is
isomorphic either to $C_5$ or to $K_{\frac{n-1}{2},\frac{n+1}{2}}$.

\bigskip

\section{The main result}\

\begin{theorem}
\label{mytheorem5} Let $G$ be a graph on $n\geq 3$ vertices
satisfying the conditions:
\begin{enumerate}
    \item [(a)] $\delta(G)\geq \frac{n-1}{2}$;
    \item [(b)] $G$ has no perfect matching;
    \item [(c)] $G$ is triangle-free.
\end{enumerate}
Then $G$ is either $C_5$ or $K_{\frac{n-1}{2},\frac{n+1}{2}}$.
\end{theorem}

\begin{proof} First of all, note that by (a) and the well-known result due to Ore \cite{b5},
we have that $G$ has a hamiltonian path. Now, since by (b) $G$ has
no perfect matching, we deduce that $n$ is odd.

Let us show that $\Delta(G)\leq \frac{n+1}{2}$. Suppose that
$\Delta(G)\geq \frac{n+3}{2}$, and let $u$ be a vertex of maximum
degree in $G$. Let $v_{1},\ldots,v_{k}$ be the neighbours of $u$.
Note that $k\geq \frac{n+3}{2}$. Now, let $u_{1},\ldots,u_{l}$ be
the remaining vertices of $G$. Note that

\begin{center}
$l=n-k-1\leq n-\frac{n+3}{2}-1=\frac{n-5}{2}$.
\end{center}

Now, (a) implies that $d_{G}(v_{1})\geq \frac{n-1}{2}$. Since $l\leq
\frac{n-5}{2}$, there should be an edge $v_{1}v_{i}\in E(G)$, where
$2\leq i \leq k$. This contradicts (c), since $u, v_{1}, v_{i}$ form
a triangle of $G$. Thus, $\Delta(G)\leq \frac{n+1}{2}$.

Next, let us show that if $\Delta(G)= \frac{n+1}{2}$, then $G$ is
isomorphic to $K_{\frac{n-1}{2},\frac{n+1}{2}}$. Suppose that
$\Delta(G)= \frac{n+1}{2}$, and let $u$ be a vertex of maximum
degree in $G$. Let $v_{1}\ldots,v_{k}$ ($k=\frac{n+1}{2}$) be the
neighbours of $u$, and let $u_{1},\ldots,u_{l}$ be the remaining
vertices of $G$. Note that

\begin{center}
$l=n-k-1= n-\frac{n+1}{2}-1=\frac{n-3}{2}$.
\end{center}

Let us show that for each $1\leq i \leq k$ and $1\leq j \leq l$,
$v_{i}u_{j}\in E(G)$. Suppose not, that is, assume that for some $i$
and $j$, $v_{i}u_{j}\notin E(G)$. Since by (a) $d_{G}(v_{i})\geq
\frac{n-1}{2}$, there is an edge $v_{i}v_{p}\in E(G)$, which is a
contradiction, since $u, v_{i}, v_{p}$ form a triangle of $G$. Thus,
any $v_{i}$ ($1\leq i \leq k$) is adjacent to any $u_{j}$ ($1\leq j
\leq l$). Note that by (c), there can be no edge among vertices
$v_{1},\ldots,v_{k}$. Also, note that there is no edge among
$u_{1},\ldots,u_{l}$, since $d_{G}(u_{j})\geq k=\frac{n+1}{2}$ and
$\Delta(G)= \frac{n+1}{2}$ by assumption. Now, it is not hard to see
that $G$ is isomorphic to $K_{\frac{n-1}{2},\frac{n+1}{2}}$.

Thus, it remains to consider the case $\Delta(G)= \frac{n-1}{2}$,
and to show that $G$ is isomorphic to $C_{5}$. Suppose that
$\Delta(G)= \frac{n-1}{2}$. (a) implies that $G$ is an $r$-regular
graph of degree $r=\frac{n-1}{2}$. Since $n$ is odd, we have that
$r$ is even. Let us show that $r=2$. Suppose that $r\geq 4$. Choose
any edge $uv\in E(G)$, and let $u_{1},\ldots,u_{r-1}$ and
$v_{1},\ldots,v_{r-1}$ be the other ($\neq v$ and $\neq u$)
neighbours of $u$ and $v$, respectively. Note that (c) implies that
$\{u_{1},\ldots,u_{r-1}\}\cap \{v_{1},\ldots,v_{r-1}\}=\emptyset$.
Let $w$ be the remaining vertex of $G$.

Suppose that $w$ is adjacent to $k$ vertices from
$\{v_{1},\ldots,v_{r-1}\}$, where $1\leq k \leq r-1$. We can assume
that these vertices are $v_{1},\ldots,v_{k}$. Then $w$ must be
adjacent to $r-k$ vertices from $\{u_{1},\ldots,u_{r-1}\}$. Again,
we can assume that these vertices are $u_{1},\ldots,u_{r-k}$.

Note that since $r\geq 4$, we have that either $k-1<r-2$ or
$r-1-k<r-2$. Suppose that $k-1<r-2$. Consider the vertex $v_{1}$.
Note that if $v_{1}$ is adjacent to one of $v_{2},\ldots,v_{r-1}$,
then we will have a triangle in $G$ contradicting (c), therefore, we
can assume that $v_{1}$ is adjacent to none of
$v_{2},\ldots,v_{r-1}$. Since $k-1<r-2$, there is an edge
$v_{1}u_{j}\in E(G)$, $1\leq j \leq r-k$. Now, note that $w, v_{1}$
and $u_{j}$ form a triangle, which contradicts (c).

Similarly, suppose that $r-1-k<r-2$. Consider the vertex $u_{1}$.
Note that if $u_{1}$ is adjacent to one of $u_{2},\ldots,u_{r-1}$,
then we will have a triangle in $G$ contradicting (c), therefore, we
can assume that $u_{1}$ is adjacent to none of
$u_{2},\ldots,u_{r-1}$. Since $r-1-k<r-2$, there is an edge
$u_{1}v_{j}\in E(G)$, $1\leq j \leq k$. Now, note that $w, u_{1}$
and $v_{j}$ form a triangle, which contradicts (c).

Thus, $r=2$. Since $n=2r+1$, we imply that $G$ is isomorphic to
$C_{5}$. The proof of the theorem is completed.
\end{proof}\

\end{document}